\documentclass[letterpaper, 10pt, conference]{ieeeconf}      

\IEEEoverridecommandlockouts
\usepackage{verbatim}
\usepackage{epsfig}
\usepackage{amsmath}
\usepackage{amsthm}
\usepackage{amssymb}
\usepackage{psfrag}
\usepackage[T1]{fontenc}
\usepackage[ruled]{algorithm2e}

\usepackage[usenames,dvipsnames]{pstricks}
\usepackage{pst-grad} 
\usepackage{pst-plot} 
\usepackage{etoolbox}
\usepackage{subfig}
\usepackage{url}
\usepackage[noadjust]{cite}  

\makeatletter
\patchcmd{\@makecaption}
  {\scshape}
  {}
  {}
  {}
\makeatother
\DeclareMathOperator*{\minimize}{minimize}

\newtheorem{theorem}{Theorem}
\newtheorem{lemma}{Lemma}

\newtheorem{assump}{Assumption}

\newtheorem{proposition}{Proposition}

\newcommand{\D}{D_c}
\newcommand{\Dl}{D_\ell}
\newcommand{\De}{D_e}

\usepackage{dsfont}

\title{\LARGE \bf
On Mitigating the Uncertainty in Renewable Generation in Distribution Microgrids
}

\author{Arnab Dey$^{1}$, Vivek Khatana$^{1}$, Ankur Mani$^{2}$  and Murti V. Salapaka$^{1}$
\thanks{This work is supported by Advanced Research Projects Agency-Energy OPEN through the project titled "Rapidly Viable Sustained Grid" via grant no. DE-AR0001016.}
\thanks{$^{1}$ Arnab Dey \{{\tt\small dey00011@umn.edu}\}, Vivek Khatana \{{\tt\small khata010@umn.edu}\}, Murti V. Salapaka\{{\tt\small murtis@umn.edu}\} are with Department of Electrical and Computer Engineering, University of Minnesota, Twin Cities, USA, and $^{2}$ Ankur Mani \{{\tt\small amani@umn.edu\}} is with Department of Department of Industrial and Systems Engineering, University of Minnesota, Twin Cities, USA,
}
}

\begin{document}

\maketitle
\thispagestyle{empty}
\pagestyle{empty}

\begin{abstract}

In this article, we focus on the problem of mitigating the risk of not being able to meet the power demand, due to the inherent uncertainty of renewable energy generation sources in microgrids. We consider three different demand scenarios, namely meeting short-time horizon power demand, a sustained energy demand and a scenario where the power demand at a prescribed future time has to be met with almost sure guarantee with power generation being stochastic and following dynamics governed by geometric Brownian motion. For each of these scenarios we provide solutions to meet the electrical demand. We present results of numerical experiments to demonstrate the applicability of our schemes. 
\end{abstract}
\section{INTRODUCTION}
In recent times, there is a mounting interest towards the generation and utilization of clean renewable energies owing to the adverse environmental effects and fast depletion of traditional energy sources \cite{shrestha1998study,miranda2015holistic}. The advent of microgrids has provided a flexible framework for the interconnection of renewable energy sources (RES) like solar photo-voltaic (PV) systems and wind energy systems \cite{parhizi2015state}. However, the integration of renewable energy sources introduces uncertainty of  meeting the electricity demand. In particular, due to the uncertain and intermittent nature of the renewable energy sources, maintaining the balance between power supply and demand can become challenging if extra measures such as ancillary services are not present \cite{irena, yang2018battery}. Moreover, climate related catastrophic events are increasing in frequency and magnitude \cite{kishore2018mortality,campbell2012weather}; here, microgrid operation of critical infrastructures such as hospitals, powered partly by renewable energy sources provides an attractive solution. However, for such critical infrastructures, it is important to guarantee needed power and thus managing uncertainty of renewable energy sources needs to be addressed. 

Due to the stochastic nature of the solar radiation and wind, probabilistic approaches are used to model the renewable power output \cite{verdejo2016stochastic,olsson2010modeling, dong2016application,salameh1995photovoltaic}. Several tasks, such as electrification of remote areas and recovery from natural disasters, require hybrid renewable energy systems (HRES) to be operated in an islanded mode, where either the grid has become unreliable or is not available. Here, optimal allocation of renewable sources and ancillary battery energy storage systems (BESS) is desired \cite{faccio2018state, birnie2014optimal,olatomiwa2016energy}. Many researchers have proposed optimization techniques focusing on overall investment and operational cost reduction \cite{zhang2016optimal,maleki2016optimal,geem2012size,ghaffari2015energy,hedman2006comparing}. However, as the BESS and RES become economically viable, primarily due to technological improvements and energy policy enforcement, a focus on reliability of meeting power demand along with traditional focus on cost optimization is needed.

To this end, in this article we focus on microgrids sourced by renewables. The renewable energy generation unit (ReGU), possibly consisting of solar and wind, has a variable power output which results in uncertainty in the total power that can be supplied to the loads. Here, we address the problem of meeting the electricity demand of the loads using ReGUs, where batteries are used to mitigate the uncertainty inherent in ReGUs. To capture different scenarios of electricity demand we consider three different situations, (i) power demand scenario where the instantaneous power demand of the loads is to be met over a short-time horizon where optimality is sought with respect to statistical measures, (ii) an energy demand problem where a certain amount of energy demanded has to be provided with guarantees of optimality, (iii) a scenario where the uncertain ReGUs are required to supply the power to the loads at a future time-instant $T_f$ with an almost sure guarantee, using batteries allocated optimally.    
For each of these three problems we provide solutions for ReGUs that minimize the risk of not being able to meet the electricity demand of the loads due to their power output variability. Numerical simulations to illustrate the applicability of our schemes corroborate the analytical/algorithmic claims.

\noindent The major contribution of this paper is threefold:

\noindent (i) We propose a stochastic optimization model to meet short-time power demand with minimum variation in renewable generation, addressing inherent uncertainties of various renewable energy sources.

\noindent(ii) In contrast to many existing optimization techniques which primarily focus on investment cost optimization of hybrid renewable energy systems to decide installed capacity before commencement of the renewable project, we provide a solution to the problem of meeting power demand in real-time with \textit{almost sure} guarantee given the stochastic nature of the renewable energy sources. Such a guarantee is essential for applications which are critical. We also provide a policy of how to optimally utilize the renewable generation and the battery storage such that the demand is met without overproduction or underproduction as well. We remark that such a solution is pertinent for supporting critical infrastructure and to the best of the authors knowledge is missing from existing state of the art.

\noindent(iii) We present a strategy to find minimum required battery reserve to meet the constant power demand throughout a time interval which minimizes the expected energy mismatch between combined generation from stochastic renewable sources and battery and demand throughout the time interval. Unlike the contributions (i) and (ii), (iii) provides meeting an energy demand instead of a power demand.

The rest of the paper is organized as follows: We provide the problem formulation for the three electricity demand scenarios in Section~\ref{sec:probform}. Then, we present the proposed schemes in Section~\ref{sec:schemes} along with their analysis and discussion on implementation. We also give characterization on how the proposed schemes are able to solve the corresponding risk minimization problem associated with each scenario. In Section~\ref{sec:sim}, we present the results of the numerical experiments pertaining to these scenarios and provide a discussion on suitability of the proposed schemes. Section~\ref{sec:conclusion} provides the concluding remarks.

\section{Problem Formulation}\label{sec:probform}
\begin{figure}[h!]
    \centering
    \includegraphics[scale=0.5]{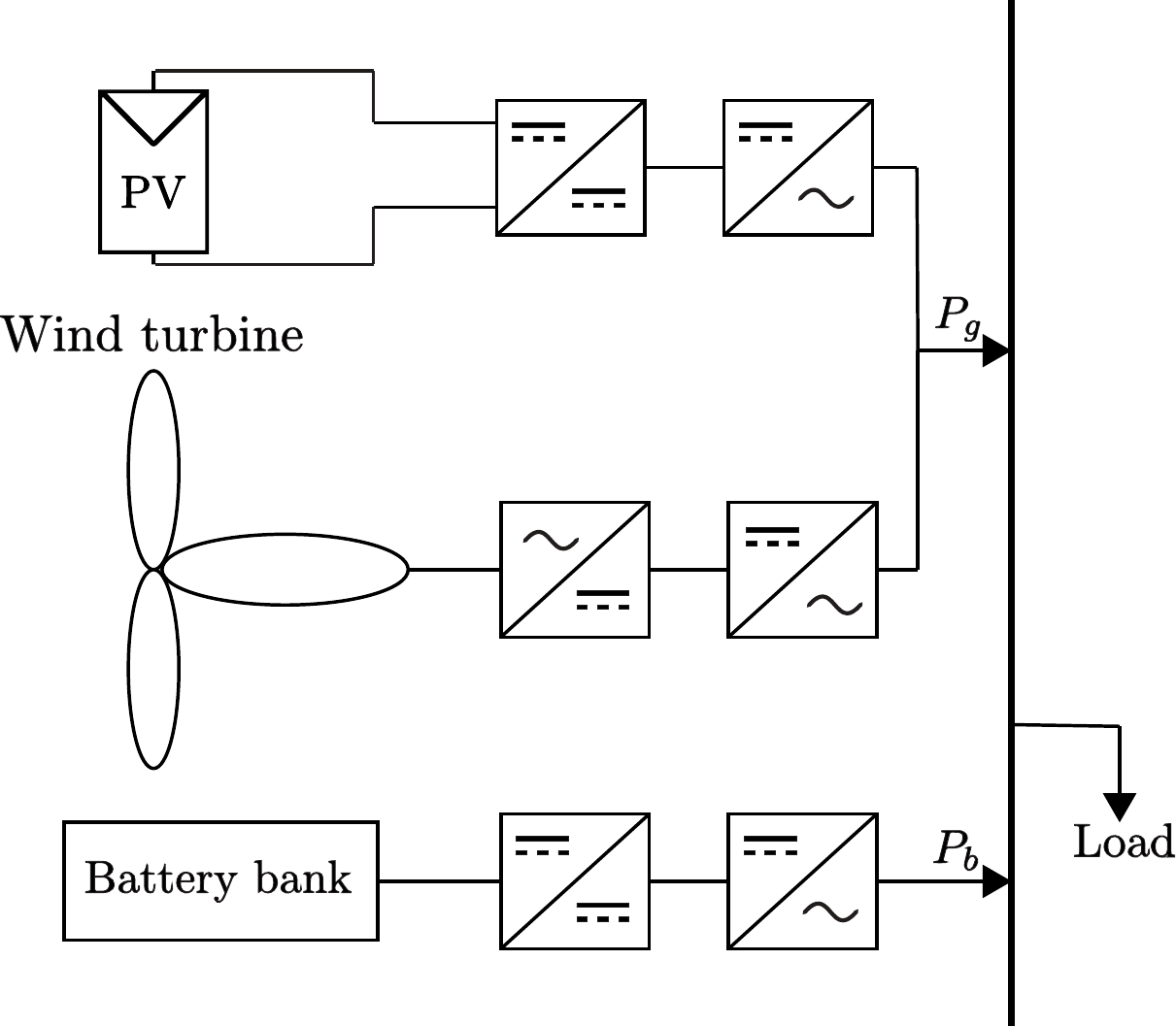}
    \caption{Schematic of PV/WT/Battery-based hybrid system}
    \label{fig:hres_sch}
\end{figure}
Schematic diagram of a typical HRES (Wind-PV-Battery) is shown in Fig.~\ref{fig:hres_sch}. Given the uncertainties in renewable power generation, we introduce the problem of meeting the load demand under the three scenarios.
We begin with a short-term power demand problem. 
\begin{subsection}{Short-term Power Demand Problem}
Here we consider the problem of a microgrid which has access to $n$ ReGU's with different renewable assets. The microgrid has no recourse to batteries. We denote the power generated by the $i^{th}$ ReGU as $e_i$ which is modeled as a normal random variable with mean $\mu_i$ and variance $\sigma_i^2$, that is, $e_i \sim \mathcal{N}(\mu_i,\,\sigma_i^{2})\,$. We further let $\mathbf{e} = [e_1 \ e_2\ \dots\ e_n]^T$, $\mu = [\mu_1\ \mu_2\ \dots\ \mu_n]^T $, with $\mu_i = \mathop{\mathbb{E}(e_i)}$, and $\mathbf{Cov}(e) = \mathbf{R}_e$, where $\mathbf{Cov}(x)$ denotes the covariance of a random vector $x$. The objective is to determine an optimal combination of ReGUs, that minimizes the variability in the generation and ensures the availability of $\Dl$ units of power to the loads. The total amount of power generated, $P_g$, by the ReGUs in the microgrid is:
\begin{align}\label{eq:Pg}
    P_g =  \textstyle \sum_{i=1}^n \alpha_i e_i,
\end{align} 
where, each ReGU $i$ provides $\alpha_i e_i$ units of power, where $\alpha_i \geq 0$. Notice, that $P_g$ is also a normal random variable with mean $
\overline{\mu} = \mu^T \alpha = \sum_{i=1}^n \alpha_i \mu_i$, and variance $\overline{\sigma}  = \boldsymbol{\alpha}^T \mathbf{R}_e \boldsymbol{\alpha}$.
The following optimization problem encapsulates the objective:
\begin{align}\label{eq:optineq}
    \minimize_{\boldsymbol{\alpha}} \ \ & \textstyle\frac{1}{2} \boldsymbol{\alpha}^T \mathbf{R}_e \boldsymbol{\alpha}  \\
   \text{subject to} \ \  & \mu^T\boldsymbol{\alpha} \geq \Dl \nonumber \\
   & 0 \leq \alpha_i \leq 1, \ \text{for all} \ i \in \{1,2, \dots, n\}\nonumber.
\end{align}

The constraints on the weights, $\alpha_i$ can be changed slightly to obtain a modified problem in the following form:
\begin{align}\label{eq:opteq}
    \minimize \ \ & \textstyle\frac{1}{2} \boldsymbol{\alpha}^T \mathbf{R}_e \boldsymbol{\alpha} \nonumber \\
    \text{subject to} \ \ & \mu^T\boldsymbol{\alpha} \geq \Dl \\
    & \mathds{1}^T \boldsymbol{\alpha} = 1 \nonumber \\
   & \boldsymbol{\alpha} \succeq 0, \nonumber
\end{align}
here, $\mathds{1}$ denotes a column vector with all entries equal to 1. In the Subsection~\ref{sec:KKTsoln}, we provide a strategy to solve~(\ref{eq:opteq}) and obtain the optimal weights $\alpha^*_i, \ i =1, 2, \dots,n$. Next, we present the energy demand problem.
\end{subsection}

\begin{subsection}{Energy Demand Problem}
Consider, a scenario where the load demand to be met in the microgrid is $\De$ units of power throughout up to a time horizon $t = T$, starting at $t = 0$. 

Here the microgrid has access to a single ReGU with output power $P_g(t)$ which is stochastic and follows a geometric brownian motion (GBM) described by:
\begin{align}\label{eq:dpower_gbm}
   \text{d}P_g(t) = \mu_g P_g(t) \text{d}t + \sigma_g P_g(t) \text{d}W_t,
\end{align}
where the constants, $\mu_g$ and $\sigma_g$, are the percentage drift and percentage volatility terms respectively, and $\text{d}W_t$ denotes a Wiener Process. Here we assume that $P_g(t) > 0,\ \text{for all}\ t \geq 0$. The microgrid has to determine $K_{batt}$ number of battery units, each capable of generating $P_b$ units of power, in an optimal manner such that, in combination with ReGU's power generation, a load of $\De$ units is sustained over $T$ units of time. Thus an energy demand of $\De T$ units has to be met.

The power mismatch, $\nu(t)$ at any $t \in [0, T]$ is:
\begin{align}\label{eq:powerdeficit}
    \nu(t) = P_g(t) + K_{batt}P_b - \De.
\end{align}
The main objective is to determine $K_{batt}$ that solves:
\begin{align}\label{eq:battblockopt}
        \minimize_{K_{batt}} \textstyle\int_0^{T} &\mathbb{E}\big[(P_g(t) + K_{batt} P_b - \De)^2\big] \text{d}t \nonumber \\
        &\text{subject to} \ K_{batt} \geq 0,\ P_g(t) \mbox{ satisfies } (\ref{eq:dpower_gbm}).
\end{align}
As will be seen later the problem~(\ref{eq:battblockopt}) becomes trivial or ill-posed if the horizon $T$ is too large; note that $\mu_g$ and $\sigma_g$, in the GBM model of $P_g(t)$, need to reflect adequate time-scales in relation to the time horizon $T$ being considered. 

In Section~\ref{sec:energysoln} we make precise the concerns raised here and present a solution to provide insights on the choice of the horizon $T$ to be chosen. For a larger horizon, the solution is pertinent only if the larger horizon is subdivided into smaller intervals based on time scale associated with $\mu_g$ and $\sigma_g$.
\end{subsection}

\begin{subsection}{Future-time Power Demand Problem}\label{sec:futurepower}
Here we consider the scenario where a microgrid consists of both critical and non-critical loads. The critical infrastructure has access to a renewable generation unit that produces power $P_g(t)$ given by~(\ref{eq:dpower_gbm}) and it postulates a need of $\D$ units of power at a future time $T_f$.
In case $P_g(T_f) \geq \D$, the microgrid can ensure the sustained operation of the critical loads. However, as the generation is uncertain the microgrid must take measures for the case when $P_g(T_f) < \D$ to avoid risk of not being able to sustain the critical loads. To meet this objective the microgrid enters in a contract with the Renewable Generation Farm (ReGF) that contains a pool of generation units (batteries and renewable generation units) which enables the microgrid to get $\D$ units of power at a future time $T_f$ if required. The ReGF maintains a portfolio of renewable generation units and battery blocks which will be utilized in case the microgrid is not able to meet the critical demand $\D$. For the ReGF, let the number of renewable generation units at any time $t$ be denoted as $a(t)$ and the number of battery blocks as $b(t)$. We will focus on the time evolution of this portfolio (evolution of $a(t)$ and $b(t)$) of power sources with the time starting from the point of entering into the contract designated as $t = 0$ till the time $t = T_f$ at which the critical power is needed. At time $t=T_f$ there are two possibilities:
\begin{enumerate}
    \item $P_g(T_f) \geq \D$ : This means that microgrid has enough power to supply the critical load demand $\D$. In this case, the microgrid has surplus power of $P_g(T_f) - \D$ units which can be used to supply power to other non-critical loads and there is no need to get additional power from the ReGF.
    \item $P_g(T_f) < \D$ : In this case the microgrid is not able to meet the critical load demand by itself and the ReGF will provide $\D$ units of power to the microgrid to support at least critical demand.
\end{enumerate}
The problem from the perspective of the ReGF is given as:
\begin{itemize}
\item Determine initial number, $b(0),$  of battery blocks and number, $a(0)$, of renewable generation units, and

\item Determine the number, $b(t)$ of battery blocks and number $a(t)$ of renewable generation units based on $P_g(t)$,  such that $\text{d}a(t)P_g+\text{d}b(t)P_b=0$ for $t\in [0,T_f],$  to ensure {\it almost surely} that
\begin{enumerate}
\item[(i)] $a(T_f)P_g(T_f)+b(T_f) P_b = \D$ \ if $P_g(T_f) < \D$,
\item[(ii)] $a(T_f)P_g(T_f)+b(T_f) P_b = 0$ \ \ \ if $P_g(T_f) \geq \D$.
\end{enumerate}
\item Determine the non-critical load demand that can be served while ensuring the critical demand of $\D$ units at $T_f$ is met \textit{almost surely}.
\end{itemize}

\noindent Here $\text{d}a(t)P_g(t)+\text{d}b(t)P_b=0$ for $t\in [0,T_f]$ ensures that the power change due to changes in the number of battery units and the renewable generation units is zero; thus the power change is only due to the change in the renewable generation power. Thus, after the initial allocation $a(0)$ and $b(0),$ the ReGF, at a future time $t$, can change the number of battery units but has to ensure that the power change is compensated by exchanging with the renewable power generation units.

Note that finding the amount of initial battery blocks and the initial renewable generation units is essential for the ReGF as having a lesser number of batteries and generation units has a risk of not being able to provide $\D$ units of power at time $T_f$, whereas provisioning more may result in excess energy produced at $t=T_f$ that will lead to a loss of revenue to the ReGF as only $\D$ units of power is required by the microgrid. Subsection~\ref{sec:BSsoln}, presents the proposed strategy in this article to find the solution for the future-time power demand problem. 

\end{subsection}

\begin{section}{Solution Methodologies}\label{sec:schemes} 
We will treat each of the scenarios outlined in Section~\ref{sec:probform} individually in the coming subsections. We start with the short-term power demand problem.
\subsection{A Scheme for Short-term power Demand Problem}\label{sec:KKTsoln}
Here, we present a scheme to solve~(\ref{eq:opteq}). Without loss of generality we assume that the optimal allocation $\alpha^*$ exists. The Lagrangian associated with~(\ref{eq:opteq}) is: 
\begin{align}
   L = \textstyle\frac{1}{2}\boldsymbol{\alpha}^T\mathbf{R}_e\boldsymbol{\alpha} - \lambda^T\boldsymbol{\alpha} + \delta [\mathds{1}^T\boldsymbol{\alpha} - 1] + \gamma [\Dl - \mu^T \boldsymbol{\alpha}],
\end{align}
where, $\lambda \succeq 0$, $\gamma \geq 0$ and $\delta \in \mathbb{R}$ are the Lagrange multipliers. Writing the KKT conditions for the above Lagrangian:
\begin{align}
    &\mathbf{R}_e \boldsymbol{\alpha} + \delta \mathds{1} - \gamma \mu = \lambda, \  \lambda \otimes \boldsymbol{\alpha} = 0 ,\gamma [\Dl - \mu^T \boldsymbol{\alpha}] = 0, \label{eq:kkt1}\\
    & \hspace{0.2in}\boldsymbol{\alpha} \succeq 0, \ \lambda  \succeq 0, \gamma \geq 0, \ \mathds{1}^T\boldsymbol{\alpha} = 1, \ \Dl \leq \mu^T \boldsymbol{\alpha}, \label{eq:kkt5}
\end{align}
where, $\lambda \otimes \boldsymbol{\alpha}$ denote the Hadamard (entry-wise) product of vectors $\lambda$ and $\boldsymbol{\alpha}$. Note, that finding a closed form solution is not possible in general. However, a solution to KKT system of equations~(\ref{eq:kkt1}), (\ref{eq:kkt5}) can be found using commercial solvers including CPLEX \cite{cplex}, GUROBI \cite{gurobi} and MOSEK \cite{mosek}. In a special case however it is possible to solve the above KKT system of equations and get a closed form solution. To this end, we make the following assumption on~(\ref{eq:opteq}):
\begin{assump}\label{ass:uncorr}
The random variable $e_i\sim\mathcal{N}(\mu_i, \sigma_i)$, is uncorrelated with the random variable $e_j\sim\mathcal{N}(\mu_j, \sigma_j)$ for all $i,j \in \{1, 2,\dots, n\}$. 
\end{assump}
\noindent Note that Assumption~\ref{ass:uncorr} is valid when the renewable energy sources are subjected to uncorrelated external conditions. This can happen when the renewable energy sources are placed at different geographical locations which are subjected to different short-term weather conditions. 

Under Assumption~\ref{ass:uncorr}, the covariance matrix $\mathbf{R}_e$ is a diagonal matrix. The objective function $\frac{1}{2}\alpha^T \mathbf{R}_e \alpha$ can be expressed in the components of $\mathbf{R}_e$ as: $\frac{1}{2}\sum_{i=1}^n\sigma_i\alpha_i^2$, where, $\sigma_i$ and $\alpha_i$ is the variance of the random variable $e_i$ and the share of ReGU $i$ in the generated power $P_g$ respectively. Equations~(\ref{eq:kkt1})-(\ref{eq:kkt5}) can be written as:
\begin{align}
    \sigma_i \alpha_i - \lambda_i + \delta - \gamma \mu_i &= 0, \ \lambda_i \alpha_i = 0, \ i = 1, 2, \dots, n, \label{eq:stationarity}\\
    & \hspace{-0.3in} \gamma \left[\Dl - \textstyle \sum_{i=1}^n \alpha_i \mu_i \right] = 0, \label{eq:comp2}\\
    \alpha_i \geq 0, \ \lambda_i & \geq 0,  \ i = 1, 2, \dots, n, \label{eq:feas1}\\
    \textstyle \sum_{i=1}^n \alpha_i = 1, \ \Dl & \leq \textstyle \sum_{i=1}^n \alpha_i \mu_i, \ \gamma \geq 0. \label{eq:feas2}
\end{align}
There are two possible cases. In the first case corresponding to an interior solution, the optimal power production is higher than $D_{\ell}$ and the dual optimal $\gamma = 0$. In the second case, the inequality constraint in~(\ref{eq:feas2}) is critical, and $\gamma > 0$. To find a closed form solution of equations~(\ref{eq:stationarity})-(\ref{eq:feas2}) we consider the two cases: 
\begin{subsubsection}{Case 1: Excess Production}
Assume, under the optimal solution, $\sum_{i=1}^n \alpha^*_i \mu_i > \Dl$. Here~(\ref{eq:comp2}) implies that $\gamma^* = 0$ and the KKT conditions reduce to:
\begin{align*}
    \sigma_i \alpha^*_i - \lambda^*_i + \delta^* = 0, \ \lambda^*_i\alpha^*_i = 0, \ i = 1,2,\dots, n, &\\
    \textstyle \sum_{i=1}^n \alpha^*_i = 1, \ \alpha^*_i \geq 0, \ \lambda^*_i \geq 0, \ i = 1,2,\dots, n.&
\end{align*}
As, $\lambda^* = [\lambda^*_1 \dots \lambda^*_n]$ acts as a slack variable it can be eliminated leaving,
\begin{align*}
    \textstyle \sum_{i=1}^n \alpha^*_i = 1, \ \alpha^*_i \geq 0, \ (\sigma_i \alpha^*_i + \delta^*)\alpha^*_i = 0, \ i = 1,\dots, n, &\\
    \sigma_i \alpha^*_i + \delta^* \geq 0, \ i = 1,\dots, n.&
\end{align*}
If $\delta^* < 0$, the last condition can only hold if $\alpha^*_i > 0$, which implies $\alpha^*_i\sigma_i + \delta^* =0$. Solving for $\alpha^*_i$ we conclude $\alpha^*_i = -\delta^* / \sigma_i$ if $\delta < 0$. If $\delta \geq 0$, it is impossible to have $\alpha^*_i > 0$ as it will violate the complementary slackness condition. Therefore, $\alpha^*_i = 0$, if $\delta^* \geq 0$. Thus for all $i = 1,\dots,n$, 
\begin{align}\label{eq:op_sol}
\alpha^*_i = 
    \begin{cases}
        - \delta^* / \sigma_i, & \ \text{if} \ \delta^* < 0 \\
        0, & \ \text{if} \ \delta^* \geq 0,
    \end{cases}
\end{align}
Note, that since $\sum_{i=1}^n \alpha^*_i = 1$ we cannot have $\delta^* \geq 0$.  Therefore, substituting~(\ref{eq:op_sol}) in the primal feasibility condition, $\sum_{i=1}^n \alpha^*_i = 1$, we get, $\delta^* = -1/\sum_{i=1}^n \frac{1}{\sigma_i}$.

Therefore, the optimal $\alpha^*_i$ is given as $\alpha^*_i = (1/\sigma_i)(1/\sum_{i=1}^n \frac{1}{\sigma_i})$. We call this an \textit{Excess Production (EP)} solution. Let $\boldsymbol{\alpha^*} = [\alpha_1^* \dots \alpha_n^*]$. The solution $\boldsymbol{\alpha^*}$ should be a feasible solution satisfying $\sum_{i=1}^n \alpha_i^* \mu_i > \Dl$. If it holds, then $\boldsymbol{\alpha^*}$ is the solution of~(\ref{eq:opteq}) and no further work is required. If this is not the case then we know that the constraint is $\sum_{i=1}^n \alpha^*_i \mu_i = \Dl$ at the optimal solution and we have the following case.  
\end{subsubsection}
\begin{subsubsection}{Case 2: Critical Production} Assume, for the optimal solution, $\sum_{i=1}^n \alpha^*_i \mu_i = \Dl$. Here, $\gamma$ in the complementarity condition is not $0$. The modified KKT conditions are: 
\begin{align*}
    \sigma_i \alpha^*_i - \lambda^*_i + \delta^* - \gamma^* \mu_i= 0, \ \lambda_i\alpha^*_i = 0, \ i = 1,2,\dots, n, &\\
    \textstyle\sum_{i=1}^n \alpha^*_i = 1, \ \alpha^*_i \geq 0, \ \lambda^*_i \geq 0, \ i = 1,2,\dots, n.&
\end{align*}
Eliminating the slack variable $\lambda^*$ as earlier we get,
\begin{align*}
    \textstyle\sum\limits_{i=1}^n \alpha^*_i = 1, \ (\sigma_i \alpha^*_i + \delta^* - \gamma^* \mu_i)\alpha_i = 0, i = 1,\dots, n, &\\
    \sigma_i \alpha^*_i + \delta^* - \gamma^* \mu_i \geq 0, \ \alpha^*_i \geq 0, i = 1,\dots, n. &
\end{align*}
If $\delta^* < \gamma^* \mu_i$ the last equation implies $\alpha^*_i > 0$, which gives $\alpha^*_i = \frac{1}{\sigma_i}[\mu_i \ -1][\gamma^* \ \  \delta^*]^T$. If $\delta^* \geq \gamma^* \mu_i$, then we get $\alpha^*_i = 0$ by the complementary slackness condition. Thus we have,
\begin{align}\label{eq:pb_sol}
\alpha_i = 
    \begin{cases}
        \frac{1}{\sigma_i}[\mu_i \  -1][\gamma^*\ \  \delta^*]^T, & \ \text{if} \ \delta^* < \gamma^* \mu_i \\
        0, & \ \text{if} \ \delta^* \geq \gamma^* \mu_i,
    \end{cases}
\end{align}
or, $\alpha^*_i = \max\{0,\frac{1}{\sigma_i}[\mu_i \ -1][\gamma^* \ \  \delta^*]^T\}$. From the primal feasibility condition, $\sum_{i=1}^n \alpha^*_i = 1$, gives
\begin{align}\label{eq:pb}
   \textstyle\sum_{i=1}^n\max \left\{0,\textstyle\frac{1}{\sigma_i}[\mu_i \  -1][\gamma^* \ \  \delta^*]^T\right \} = 1.
\end{align}
Solving the univariate optimization problem in $[\gamma^* \ \  \delta^*]^T$ gives the solution to the original problem. The solution of~(\ref{eq:pb}) can be found using a water filling algorithm \cite{he2013water}.  We term this solution a \textit{Critical Production (CP)} solution.  We present the procedure to solve~(\ref{eq:opteq}) in Algorithm~\ref{alg:kktsoln}.
\begin{algorithm}[h]
    \SetKwBlock{Input}{Input:}{}
    \SetKwBlock{Initialize}{Initialize:}{}
    \SetKwBlock{STEPONE}{STEP 1:}{}
    \SetKwBlock{STEPTWO}{STEP 2:}{}
    \SetKwBlock{Repeat}{Repeat for $ k = 1,2, \dots$}{}
    \newcommand{\inparallel}{\textbf{In parallel}}
    \Input{ Mean generation vector $ \mu = [\mu_1\ \dots\ \mu_n ]$; \\
    Variance vector of the ReGUs $\sigma = [\sigma_1\ \dots\ \sigma_n]$;}\vspace{0.1in}
    Compute \textit{EP} solution $\boldsymbol{\alpha}^*$; \\\vspace{0.05in}
      \uIf {$\sum_{i=1}^n \alpha_i^* \mu_i > \Dl$} {\vspace{0.05in} EP is the solution to~(\ref{eq:opteq}); 
                }
     \uElse { 
             $\boldsymbol{\alpha}^*$ is the \textit{CP} solution~(\ref{eq:pb_sol}), obtained by \hspace*{-0.1in} solving~(\ref{eq:pb})
            }
    
      \caption{Solution Procedure for solving~(\ref{eq:opteq})}
        \label{alg:kktsoln}
\end{algorithm}
\end{subsubsection}

\begin{subsection}{Proposed Battery Reserve Design}\label{sec:energysoln}
In this subsection, we solve for the battery reserve as stated in problem~(\ref{eq:battblockopt}). 
\noindent From~(\ref{eq:dpower_gbm}) it follows that the expected value and variance of $P_g(t)$ are given by the following expressions:
\begin{align}
        \mathbb{E}[P_g(t)] &= P_g(0)e^{\mu_g t},\\
        \mathbf{Var}[P_g(t)] &= P_g^2(0)e^{2\mu_g t}(e^{\sigma_g^2t}-1).
\end{align}
Note that the variance of $P_g(t)$ grows exponentially from zero. Thus the optimal solution, $K_{batt}$ in~(\ref{eq:battblockopt}) will still incur a large mismatch from the desired power if the horizon $T$ is very large where $\sigma_g$ does not reflect the volatility associated with the time scale of $T$.

To this end, we propose a scheme in which we divide the time interval $[0, T]$ into sub-intervals $[0,t_1], [t_1,t_2],\dots, [t_j,T]$, based on time scale associated with given $\mu_g$ and $\sigma_g$, for some finite natural number $j$ and solve~(\ref{eq:battblockopt}) for these sub-intervals incrementally. Let $\Delta t_i$ to be the length of the $i^{th}$ sub-interval $[t_{i-1},t_i]$, with $t_i = t_{i-1}+ \Delta t_i$. Let $K_i$ be the solution to~(\ref{eq:battblockopt}) with the interval $[0,T]$ replaced by $[t_{i-1},t_i]$. To avoid an overestimate we determine  $\Delta t_i$ such that the expectation of squared power mismatch, $\nu(t)$, over the time-interval $[t_{i-1},t_i]$ is constrained below a certain desired tolerance $\varepsilon > 0$. For a given tolerance bound $\varepsilon$, we calculate the length $\Delta t_i$ of the $i^{th}$ sub-interval by solving the following equation: 
\begin{align}\label{eq:intervallength}
\hspace{-0.09in} \varepsilon &= \textstyle\frac{1}{\Delta t_i}\int_{t_{i-1}}^{t_i} \mathbb{E}\left[(P_g(t) + K_i P_b - \De)^2|P_g(t_{i-1})\right] \text{d}t,
\end{align}
where $K_i$ is substituted with the expression given in Proposition~\ref{prop:optimal_energy}.
Once $\Delta t_i$ is determined, we solve the optimization problem~(\ref{eq:battblockopt}) for the interval $[t_{i-1},t_i]$ to find the numerical value of $K_i$. The whole process is repeated for the next sub-interval $[t_i, t_{i+1}]$ to find $\Delta t_{i+1}$ and $K_{i+1}$ until $\sum\limits_i \Delta t_i \geq T$.
\begin{proposition}\label{prop:optimal_energy}
The number of battery blocks $K_i$, for all $i = 1,2,\dots$, that solve~(\ref{eq:battblockopt}) for the $i^{th}$ sub-interval $[t_{i-1}, t_i]$ is given by: 
\begin{align*}
     K_i = \max \left \{ 0, \textstyle\frac{1}{P_b}\left[\De - \textstyle\frac{P_g(t_{i-1})e^{\mu_g t_{i-1}}}{\mu_g \Delta t_i} \left(e^{\mu_g \Delta t_i} - 1 \right) \right] \right\}.
\end{align*}
\end{proposition}
\begin{proof}
The optimization problem~(\ref{eq:battblockopt}) for $i^{th}$ sub-interval is:
\begin{align}\label{eq:subintervalopt}
        \minimize_{K_i} \textstyle\int_{t_{i-1}}^{t_i} &\mathbb{E}\left[(P_g(t) + K_{i}P_b - \De)^2 | P_g(t_{i-1})\right] \text{d}t \nonumber \\
        &\text{subject to} \ K_{i} \geq 0,\ P_g(t) \mbox{ satisfies } (\ref{eq:dpower_gbm}). 
\end{align}
Writing the Lagrangian of the above problem we have,
\begin{align*}
  L & = K_i^2 P_b^2\Delta t_i + \textstyle\int_{t_{i-1}}^{t_i} \mathbb{E}(P_g^2(t)|P_g(t_{i-1})) \text{d}t + D^2_e\Delta t_i \\
  &- 2K_i P_b \De \Delta t_i - 2\De\textstyle\int_{t_{i-1}}^{t_i} \mathbb{E}(P_g(t)|P_g(t_{i-1})) \text{d}t \\
  &+ 2K_i P_b\textstyle\int_{t_{i-1}}^{t_i} \mathbb{E}(P_g(t)|P_g(t_{i-1})) \text{d}t - \lambda K_i,
\end{align*}
where, $\lambda \geq 0$, is the Lagrange multiplier. We have the following KKT conditions:
\begin{align*}
    2K_iP_b^2\Delta t_i - & 2P_b\De \Delta t_i \\
    &+ 2P_b\textstyle\int_{t_{i-1}}^{t_i} \mathbb{E}(P_g(t)|P_g(t_{i-1}))\text{d}t - \lambda = 0, \\
    & K_i \lambda = 0, \ K_i \geq 0, \ \lambda \geq 0.
\end{align*}
Solving the above KKT system of equations we get: 
\begin{align}
    K_i^* & = \max \left \{ 0, \textstyle\frac{1}{P_b}\left[\De - \textstyle\frac{1}{\Delta t_i} \int_{t_{i-1}}^{t_i} \mathbb{E}(P_g(t)|P_g(t_{i-1}))\text{d}t \right] \right\} \nonumber \\
    & \hspace{-0.3in} = \max \left \{ 0, \textstyle\frac{1}{P_b}\left[\De - \frac{P_g(t_{i-1})e^{\mu_g t_{i-1}}}{\mu_g \Delta t_i} \left(e^{\mu_g \Delta t_i} - 1 \right) \right] \right\}. \label{eq:optKi}
\end{align}
Since,~(\ref{eq:subintervalopt}) is a convex optimization problem, therefore the solution of the KKT system of equations is the optimal solution of~(\ref{eq:subintervalopt}). This completes the proof.
\end{proof}
\noindent We summarize the proposed strategy for maintaining the battery reserve in real-time in Algorithm~\ref{alg:batteryblocks}.
\begin{algorithm}[h]
    \SetKwBlock{Input}{Input:}{}
    \SetKwBlock{Initialize}{Initialize:}{}
    \SetKwBlock{STEPONE}{STEP 1:}{}
    \SetKwBlock{STEPTWO}{STEP 2:}{}
    \SetKwBlock{Repeat}{Repeat for $ i = 1,2, \dots$}{}
    \newcommand{\inparallel}{\textbf{In parallel}}
    \Input{ tolerance parameter $\varepsilon$; \\
    $\mu_g$ and $\sigma_g$;}
    \Repeat{  
       Compute $\Delta t_i$ using~(\ref{eq:intervallength}); \ 
       Find $K_i$ using~(\ref{eq:optKi});\\
       \If {$\sum\limits_{i} \Delta t_i \geq T$}
       {\hspace*{0.05in} \textbf{break};}
       }
      \vspace{0.025in}
        \caption{Battery Reserve Computation for the \hspace*{0.815in} Energy Demand Problem}
        \label{alg:batteryblocks}
\end{algorithm}
\end{subsection}

\subsection{Future-time Power Demand Problem}\label{sec:BSsoln}
Here, we present a scheme to solve the problem of meeting the power demand of $\D$ units at a future time instant, $T_f$ introduced in the Subsection~\ref{sec:futurepower}. We provide a policy of maintaining the number of battery blocks, $b(t)$ and the number of generation units, $a(t)$ to almost-surely meet the critical demand $\D$ units at time $T_f$. The power available in the portfolio maintained by the ReGF, depends on $P_g(t)$, satisfying~(\ref{eq:dpower_gbm}) and time $t \leq T_f$, and is given by,
\begin{align}\label{eq:portfolio}
    V(P_g(t), t) = a(t)P_g(t) + b(t)P_b,
\end{align}
where, $a(t)$ is the number of generation units and $b(t)$ is the number of battery blocks in the portfolio at time $t \leq T_f$. In the subsequent development the explicit dependency of the variables on time $t$ is omitted for brevity of notations.
\begin{lemma}
Under, the constraint, $\text{d}a(t)P_g(t) + \text{d}b(t) P_b = 0$ for all $t \in [0,T_f]$ and if $a = \frac{\partial V}{\partial P_g}$, we have 
\begin{align*}
   \textstyle\frac{\partial V}{\partial t} + \frac{1}{2} \sigma_g^2 P_g^2 \textstyle\frac{\partial^2 V}{\partial P_g^2} = 0.
\end{align*}
\end{lemma}
\begin{proof}
The variation of the power generated by the each generation unit, $P_g(t)$, is governed by a GBM:
\begin{align}\label{eq:ins_prob_gbm}
    \text{d}P_g = \mu_g P_g \text{d}t + \sigma_g P_g \text{d}W_t,
\end{align}
where, $\mu_g$ and $\sigma_g$ are constants and $\text{d}W_t$ is a Wiener Process. Further, since $P_b$ is constant therefore,
\begin{align}\label{eq:dpbatt}
    dP_b = 0.
\end{align}
Now, applying the differentiation operator to the portfolio $V$:
\begin{align}\label{eq:diffV}
    \text{d}V & = a\text{d}P_g + \text{d}aP_g + b\text{d}P_b +  \text{d}bP_b.
\end{align}
Under the constraint $\text{d}a(t)P_g + \text{d}b P_b = 0$ for all $t \in [0,T_f]$,~(\ref{eq:diffV}) becomes (omitting the time dependency of the variables to make the equations legible)
\begin{align}\label{eq:dV1}
     \text{d}V &= a \text{d}P_g + b \text{d}P_b \nonumber\\
        &= a(\mu_g P_g \text{d}t + \sigma_g P_g \text{d}W_t),
\end{align}
where, the last step follows from~(\ref{eq:ins_prob_gbm}) and~(\ref{eq:dpbatt}). Applying Ito's lemma \cite{gardiner2009stochastic} and ignoring \textit{h.o.t}, 
\begin{align}
        \text{d}V & = \textstyle\frac{\partial V}{\partial t}\text{d}t  + \textstyle\frac{\partial V}{\partial P_g}\text{d}P_g + \textstyle\frac{1}{2}\frac{\partial^2 V}{\partial P_g^2}(\text{d}P_g^2) \nonumber\\
        & \hspace{-0.3in} = \textstyle\frac{\partial V}{\partial t}\text{d}t + \textstyle\frac{\partial V}{\partial P_g}(\mu_g P_g \text{d}t + \sigma_g P_g \text{d}W_t) + \textstyle\frac{\sigma_g^2 P_g^2}{2} \textstyle\frac{\partial^2 V}{\partial P_g^2}\text{d}t. \label{eq:dV2} 
\end{align}
Using,~(\ref{eq:dV1}) and~(\ref{eq:dV2}) we get,
\begin{align}
    \left[ \textstyle\frac{\partial V}{\partial t} + \mu_g P_g\textstyle\frac{\partial V}{\partial P_g} + \textstyle\frac{1}{2} \sigma_g^2 P_g^2 \textstyle\frac{\partial^2 V}{\partial P_g^2} - a\mu_g P_g \right]&\text{d}t \nonumber \\ + \left[ \sigma_g P_g\textstyle\frac{\partial V}{\partial P_g} - a\sigma_g  P_g\right]\text{d}&W_t = 0.
\end{align}
To eliminate randomness \cite{black1973pricing}, we make,\\ $[\sigma_g P_g\frac{\partial V}{\partial P_g} - a\sigma_g P_g] = 0$, that is, $a = \frac{\partial V}{\partial P_g}$. Therefore, 
\begin{equation}\label{eq:pde}
    \textstyle\frac{\partial V}{\partial t} + \textstyle\frac{1}{2} \sigma_g^2 P_g^2 \textstyle\frac{\partial^2 V}{\partial P_g^2} = 0.
\end{equation}
\end{proof}
\noindent Next we provide a solution to~(\ref{eq:pde}) under the following terminal condition:
\begin{align}\label{eq:terminal_cond}
   \hspace{-0.1in} V(P_g(T_f),T_f) = 
    \begin{cases}
    0                        & \text{if $P_g(T_f) \geq \D$}\\
    \D - P_g(T_f)           & \text{if $P_g(T_f) < \D$}.
    \end{cases}
\end{align}

\begin{theorem}\label{thm:bs_sol}
Under, the constraint, $\text{d}a(t)P_g(t) + \text{d}b(t) P_b = 0$ for all $t \in [0,T_f]$ and if $a = \frac{\partial V}{\partial P_g}$, then the provisioning policy of the generation and the battery blocks, $a(t)$ and $b(t)$ respectively, such that the terminal condition~(\ref{eq:terminal_cond}) for the portfolio $V(P_g(t),t)$ is met almost-surely is given by:
\begin{align*}
    a(t) &= -F\bigg[\textstyle\frac{\ln\left(\frac{\D}{P_g(t)}\right)-\textstyle\frac{\sigma_g^2}{2}(T_f - t)}{\sigma_g \sqrt{(T_f - t)}}\bigg],\\ 
    b(t) &= \textstyle\frac{\D}{P_b}F\bigg[ \frac{\ln\left(\frac{\D}{P_g(t)}\right)+\textstyle\frac{\sigma_g^2}{2}(T_f - t)}{\sigma_g \sqrt{(T_f - t)}}\bigg],\ t \in [0,T_f)
\end{align*}
where, $F(.)$ is the Cumulative Distribution Function of the standard Gaussian random variable $\sim \mathcal{N}(0,1)$. Further, the amount of non-critical loads that can be served is given by $(1+|a(t)|)P_g(t)$ while ensuring the critical demand of $\D$ units is met almost surely with the terminal condition for the portfolio given in~(\ref{eq:terminal_cond})
\end{theorem}
\begin{proof}
See Subsection~\ref{sec:appdx}.
\end{proof}
\end{section}

\section{Simulation Results}\label{sec:sim}
Consider a hybrid renewable energy system (HRES) with the renewable generation $P_g(t)$ with parameters for geometric brownian motion described by $\mu_g=0.1$ and $\sigma_g=0.3$.  Suppose the battery unit power $P_b=1$ kW. It is desired to provision initial quantities, $a(0)$ and $b(0)$ and devise a policy for $a(t)$ and $b(t)$ to ensure that the power demand of $25$ kW is met at time $T_f=5$ hrs. Toward addressing the problem, we employ the strategy provided by Theorem~\ref{thm:bs_sol} which is implemented in Python 3.8. A total of $300$ samples are taken between 0 to $T_f$. The number of renewable generation units and battery blocks are adjusted at an interval of 1 minute. The random variable $P_g(t)$ is realized using~(\ref{eq:dpower_gbm}). Fig.~\ref{fig:Pg_low_end} considers the scenario where a realization of $P_g(t)$ results in the generation at $T_f$ being not sufficient to meet the power demand, $\D$,  i.e. $P_g(T_f)<\D$. It is evident that the power portfolio value at time $T_f$ becomes equal to the generation deficit given by $\D-P_g(T_f)$ to ensure that $\D$ units of power demand at $T_f$ is met.
Fig.~\ref{fig:Pg_high_end} shows the scenario where a realization of the stochastic process $P_g(t)$ leads to renewable generation at time $T_f$ being more than sufficient to meet the power demand, $\D$. In this scenario, both power portfolio and number of battery block requirement become 0 at $t=T_f$ as is guaranteed by Theorem~\ref{thm:bs_sol}. Thus, irrespective of the uncertainty in $P_g$, the power demand of $\D$ units at time $T_f$ is met almost for every realization. Such a guarantee is essential for applications that are deemed critical.
\begin{figure}[h]
     \centering
     \subfloat[][$P_g(T_f) < \D$]{\includegraphics[scale=1.0,trim={0.4cm 0cm 0cm 0cm},clip]{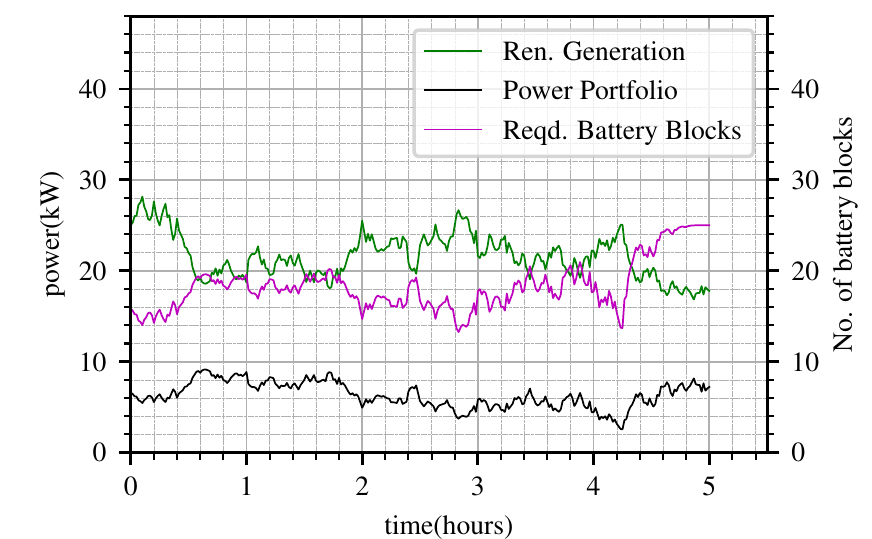}\label{fig:Pg_low_end}}\hspace{0.5cm}
     \subfloat[][$P_g(T_f) > \D$]{\includegraphics[scale=1.0,trim={0.4cm 0cm 0cm 0.07cm},clip]{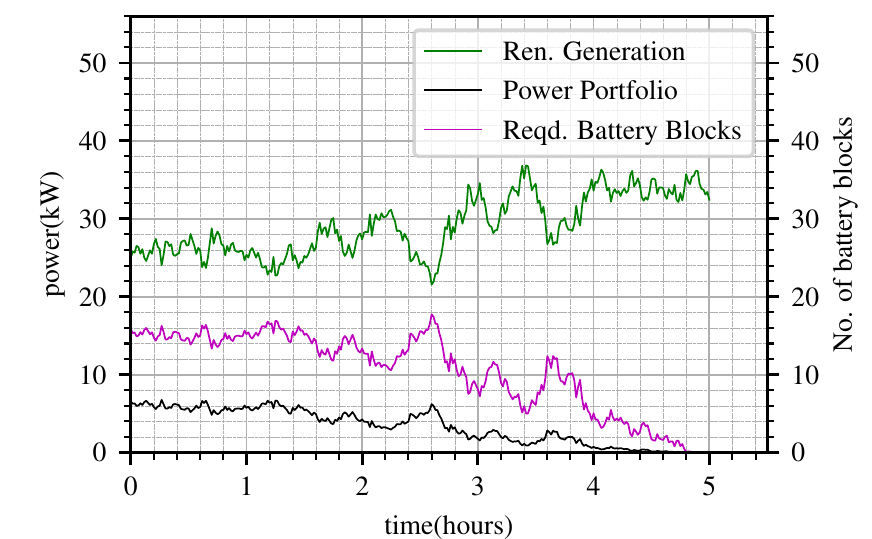}\label{fig:Pg_high_end}}
     \caption{Simulation results for two realizations of $P_g(t)$ based on~(\ref{eq:dpower_gbm}) with $\D=25$ kW, $T_f=5$ hrs}
     \label{fig:Pg_subfig_high_low}
\end{figure}
Fig.~\ref{fig:v_vs_t_same_pg} and Fig.~\ref{fig:b_vs_t_same_pg} consider the effect of time horizon length ($T_f$) on the power portfolio value and battery block requirement, respectively, at any time instant $t \in [0,T_f]$, given a realization of the stochastic process $P_g(t)$ and power demand of $\D=25$ kW at all $T_f \in \{3,4,5\}$ hours. Fig.~\ref{fig:v_vs_t_same_pg} shows that, at any time instant $t$, for same $P_g(t)$, power portfolio value is higher for longer horizon, $T_f$ if power demand of $\D=25$ kW has to be met at $T_f$ almost surely.
Similarly, Fig.~\ref{fig:b_vs_t_same_pg} shows that, to meet the power demand of $\D=25$ kW at $T_f$ with almost sure guarantee, at any time instant $t$, battery block requirement is lower for longer horizon if $P_g(t)<\D$, and higher for longer horizon if $P_g(t) \geq \D$.

\begin{figure}[h!]
    \centering
    \includegraphics[scale=1.0,trim={0.4cm 0cm 0cm 0.1cm},clip]{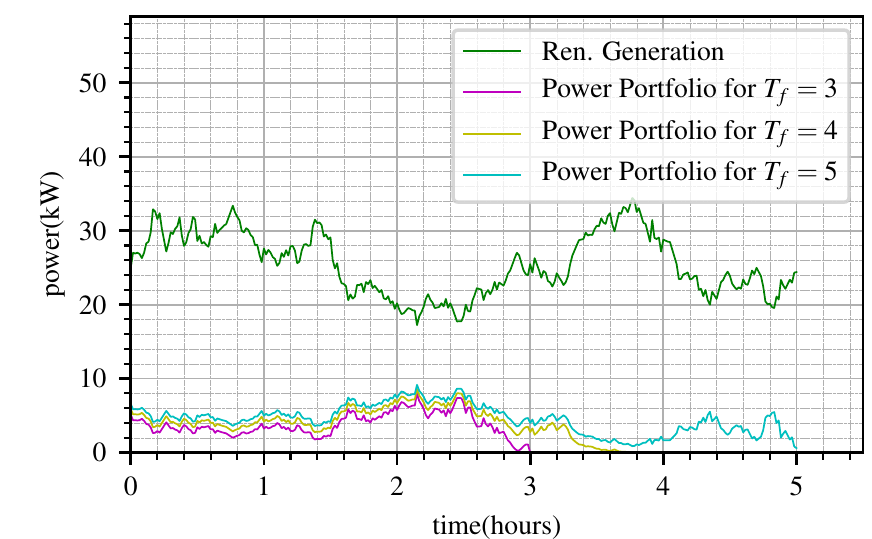}
    \caption{Variation of power portfolio with time for different $T_f$, under same $P_g(t)$ and same load demand of $\D=25$ kW at $T_f$, for a realization of $P_g(t)$ based on~(\ref{eq:dpower_gbm})}
    \label{fig:v_vs_t_same_pg}
\end{figure}
\begin{figure}[h]
    \centering
    \includegraphics[width=\columnwidth,trim={0.4cm 0cm 0cm 0cm},clip]{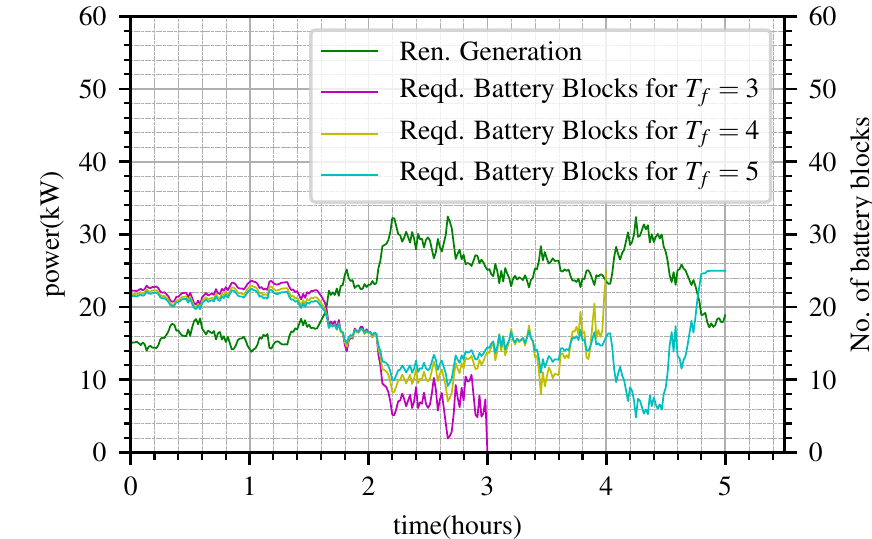}
    \caption{Variation of number of battery blocks with time for different $T_f$, under same $P_g(t)$ and same load demand of $\D=25$ kW at $T_f$, for a realization of $P_g(t)$ based on~(\ref{eq:dpower_gbm})}
    \label{fig:b_vs_t_same_pg}
\end{figure}
    
Implementation of Algorithm~\ref{alg:batteryblocks}, based on a realization of $P_g(t)$ as per~(\ref{eq:dpower_gbm}) is shown in Fig.~\ref{fig:opt_Batt} where the constant power demand of $\D=25$ kW has to be met over a time horizon of 0 to $T=5$ hours. Renewable generation profile is simulated with a total sample size of $300$ and the base power is taken as $25$ kW which is used to calculate per unit power. For an energy mismatch tolerance bound of $\varepsilon=0.01$, the time steps and number of reserve battery blocks are calculated based on~(\ref{eq:intervallength}). The result shows that the battery power follows the actual demand-generation deficit closely. While this method optimizes the battery reserve requirement, it can be augmented with the strategy proposed in Theorem~\ref{thm:bs_sol} to maintain power balance.

\begin{figure}[h!]
    \centering
    \includegraphics[scale=1.0,trim={0.4cm 0cm 0cm 0.1cm},clip]{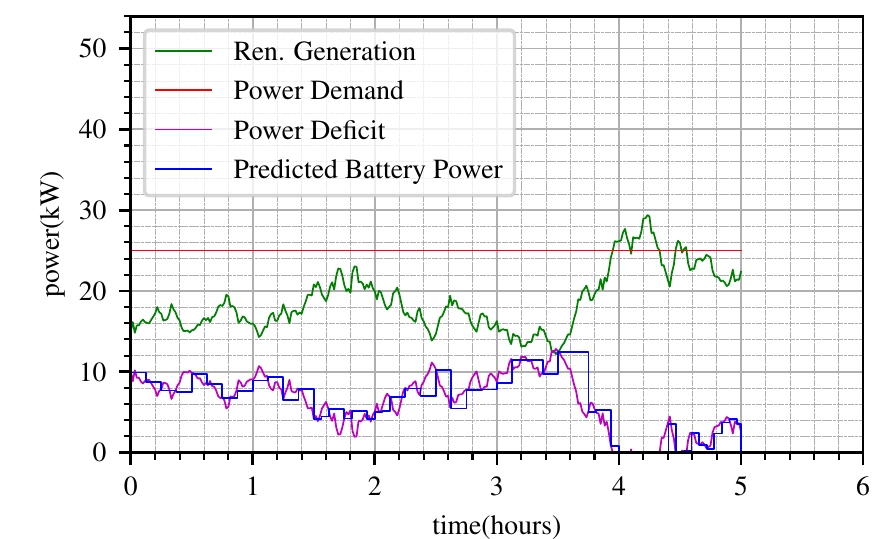}
    \caption{Optimal battery power to meet demand of $\De=25$ kW over a time horizon of $T=5$ hours for a realization of $P_g(t)$ based on~(\ref{eq:dpower_gbm})}
    \label{fig:opt_Batt}
\end{figure}

\begin{section}{Conclusion}\label{sec:conclusion}
This paper presents an optimal approach for a hybrid renewable sources and battery based system to provide power with minimum variability considering the difference in uncertainties of various renewable power sources. This solution provides a suitable combination of different renewable generation sources which minimizes the overall generation variability. Moreover, to maintain system reliability and sustainability, a strategy to guarantee load demand at a future time instant is also given. This strategy ensures that the demand is met without overproduction or underproduction by optimal allocation and utilization of renewable and battery storage. This article also presents a solution to the problem of meeting constant power demand throughout a specified time horizon by minimizing expected difference in generation and demand. Combined, these three approaches can enable an HRES to mitigate the risks associated with uncertainties of renewable energy sources.  
\end{section}
\subsection{Proof of Theorem~\ref{thm:bs_sol}}\label{sec:appdx}
Let $\tau = T_f - t$, $x = \ln(P_g)$. Therefore,
\begin{align*}
    \textstyle\frac{\partial V}{\partial t} &= - \textstyle\frac{\partial V}{\partial \tau},\ 
    \textstyle\frac{\partial V}{\partial P_g} = \frac{1}{P_g} \textstyle\frac{\partial V}{\partial x},\
    \textstyle\frac{\partial^2 V}{\partial P_g^2} = \textstyle\frac{1}{P_g^2}\Big[\frac{\partial^2 V}{\partial x^2}-\textstyle\frac{\partial V}{\partial x}\Big].
    \end{align*}
From~(\ref{eq:pde}),
\begin{equation}\label{eq:pde_1}
\begin{split}
    \textstyle\frac{\partial V}{\partial \tau} &= \frac{1}{2} \sigma_g^2 \textstyle\frac{\partial^2 V}{\partial x^2} - \textstyle\frac{1}{2} \sigma_g^2 \textstyle\frac{\partial V}{\partial x} = A \textstyle\frac{\partial^2 V}{\partial x^2} + B \textstyle\frac{\partial V}{\partial x},
\end{split}
\end{equation}
where, $A = \frac{1}{2} \sigma_g^2\ (A > 0),\ B = - \frac{1}{2} \sigma_g^2$. From~(\ref{eq:terminal_cond}), the final condition on power portfolio $V(P_g(t), t=T_f)$, or equivalently, initial condition on $V(x, \tau=0)$ is given as:
\begin{equation}\label{eq:terminal_cond_changed_var}
    V(x, 0) = 
    \begin{cases}
    0                    & \text{if $x \geq \ln(\D)$}\\
    \D - e^x            & \text{if $x < \ln(\D)$}.
    \end{cases}
\end{equation}
Let $V(x, \tau) = e^{-(\alpha x + \beta \tau)}u(x, \tau)$, where $\alpha, \beta \in \mathbb{R}$. Then,
\begin{equation*}
    \textstyle\frac{\partial V}{\partial \tau} = e^{-(\alpha x + \beta \tau)}\left[ \textstyle\frac{\partial u}{\partial \tau} - \beta u\right].
\end{equation*}
\begin{align*}
    \mbox{Similarly,} \ \ \ \ \ \ \ \ \ \ \  \textstyle\frac{\partial V}{\partial x} &= e^{-(\alpha x + \beta \tau)}\left[ \textstyle\frac{\partial u}{\partial x} - \alpha u\right]\\
    \textstyle\frac{\partial^2 V}{\partial x^2} &= e^{-(\alpha x + \beta \tau)} \left[ \textstyle\frac{\partial^2 u}{\partial x^2} - 2\alpha \textstyle\frac{\partial u}{\partial x} + \alpha^2 u\right].
\end{align*}
Therefore, using~(\ref{eq:pde_1}),
\begin{align*}
    &\textstyle\frac{\partial u}{\partial \tau} - \beta u = A\left[ \textstyle\frac{\partial^2 u}{\partial x^2} - 2\alpha \textstyle\frac{\partial u}{\partial x} + \alpha^2 u \right] + B\left[ \textstyle\frac{\partial u}{\partial x} - \alpha u \right]\\
        \implies & \textstyle\frac{\partial u}{\partial \tau} = A\textstyle\frac{\partial^2 u}{\partial x^2} + \left[ B - 2\alpha A\right]\textstyle\frac{\partial u}{\partial x} + \left[ \beta + \alpha^2 A - \alpha B\right]u.
\end{align*}
Choosing $\alpha = \frac{B}{2A}$ and $\beta = \frac{B^2}{4A}$, we get,
$\textstyle\frac{\partial u}{\partial \tau} = A\textstyle\frac{\partial^2 u}{\partial x^2}.$

\vspace*{0.2cm}
\noindent The solution of this PDE \cite{evans2010partial} is given by,
\begin{align*}
    \textstyle u(x,\tau) =& \textstyle\frac{1}{\sqrt{4\pi A\tau}}\int_{-\infty}^{\infty} \textstyle u(y,0) e^{-\frac{(x-y)^2}{4A\tau}}\text{d}y.
\end{align*}
Since, $V(y,0) = e^{-(\alpha y)}u(y, 0)$, $\alpha = \textstyle\frac{B}{2A}$, $\beta = \textstyle\frac{B^2}{4A}$,
\begin{align}\label{eq:heat_eqn}
    &\textstyle u(x,\tau) = \frac{1}{\sqrt{4\pi A\tau}} \textstyle\int_{-\infty}^{\infty} e^{-\frac{(x-y)^2}{4A\tau}} e^{\alpha y}V(y,0)\text{d}y \nonumber  \\
    &= \textstyle\frac{1}{\sqrt{4\pi A\tau}}\textstyle\int_{-\infty}^{\infty} e^{-\left[ -\frac{By}{2A} + \frac{(x-y)^2}{4A\tau}\right]}V(y,0)\text{d}y \nonumber \\
    &= \textstyle\frac{1}{\sqrt{4\pi A\tau}}\textstyle\int_{-\infty}^{\infty} e^{-\left[ \textstyle\frac{(y-B\tau-x)^2}{4A\tau} -\textstyle\frac{Bx}{2A} -\textstyle\frac{B^2\tau}{4A}\right]}V(y,0)\text{d}y \nonumber \\
    &= \textstyle\frac{1}{\sqrt{4\pi A\tau}}\textstyle\int_{-\infty}^{\infty} e^{-\left[ \frac{(y-B\tau-x)^2}{4A\tau} -\alpha x -\beta \tau\right]}V(y,0)\text{d}y \nonumber \\
    &= \textstyle\frac{1}{\sqrt{4\pi A\tau}}\textstyle\int_{-\infty}^{\infty} e^{-\left[ \frac{(y-B\tau-x)^2}{4A\tau}\right]}e^{(\alpha x + \beta \tau)}V(y,0)\text{d}y.
\end{align}
Therefore, $V(x,\tau)$ is given by,
\begin{align}\label{eq:pde_3}
    & \textstyle e^{-(\alpha x + \beta \tau)}u(x, \tau) \nonumber\\
        &= \textstyle\frac{1}{\sqrt{4\pi A\tau}}\int_{-\infty}^{\infty} e^{-\left[ \frac{(y-B\tau-x)}{2\sqrt{A\tau}}\right]^2} V(y,0)\text{d}y \nonumber\\
        &= \textstyle\frac{1}{\sqrt{2\pi\tau} \sigma_g}\int_{-\infty}^{\infty} e^{-[ (y+\textstyle\frac{\sigma_g^2}{2}\tau-x)/(\sqrt{2}\textstyle\sigma_g\tau)]^2} V(y,0)\text{d}y \nonumber\\
        &= \textstyle\frac{1}{\sigma_g \sqrt{2\pi \tau}}\int_{-\infty}^{\infty} e^{-\frac{1}{2}\Big[ \frac{(y+\frac{\sigma_g^2}{2}\tau-x)}{\sigma_g\sqrt{\tau}}\Big]^2} V(y,0)\text{d}y.
\end{align}
Applying the initial condition~(\ref{eq:terminal_cond_changed_var}) on $V(y,0)$ we get:
\begin{align}
    & V(x, \tau) =  \nonumber\\
    & \textstyle\frac{1}{\sigma_g \sqrt{2\pi \tau}}\Big[\int_{-\infty}^{\ln(\D)} e^{-\frac{1}{2}\big[ \frac{(y+\frac{\sigma_g^2}{2}\tau-x)}{\sigma_g\sqrt{\tau}}\big]^2}(\D - e^y)\text{d}y\Big]\nonumber \nonumber\\
    &= \underbrace{\textstyle\frac{\D}{\sigma_g \sqrt{2\pi \tau}}\Big[\int_{-\infty}^{\ln(\D)} e^{-\frac{1}{2}\big[ \frac{(y+\frac{\sigma_g^2}{2}\tau-x)}{\sigma_g\sqrt{\tau}}\big]^2}\text{d}y\Big]}_{I_1} \nonumber\\
    & \hspace{1cm} - \underbrace{\textstyle\frac{1}{\sigma_g \sqrt{2\pi \tau}}\Big[\int_{-\infty}^{\ln(\D)} e^{-\frac{1}{2}\big[ \frac{(y+\frac{\sigma_g^2}{2}\tau-x)}{\sigma_g\sqrt{\tau}}\big]^2}e^y\text{d}y\Big]}_{I_2}. \nonumber
\end{align}

Let, $z=\frac{1}{\sigma_g\sqrt{\tau}}(y+(\sigma_g^2/2)\tau-x) \implies \text{d}z = \frac{1}{\sigma_g\sqrt{\tau}}\text{d}y$. Therefore,
\begin{align}
    I_1 &= \textstyle\frac{\D}{\sigma_g \sqrt{2\pi \tau}}\Big[\int_{-\infty}^{\ln(\D)} e^{-\frac{1}{2}\big[ \frac{(y+\frac{\sigma_g^2}{2}\tau-x)}{\sigma_g\sqrt{\tau}}\big]^2}\text{d}y\Big]\nonumber\\
    &= \textstyle\D\Big[ \frac{1}{\sqrt{2\pi}} \int_{-\infty}^{\frac{\ln(\D)+\frac{\sigma_g^2}{2}\tau - x}{\sigma_g \sqrt{\tau}}} e^{-\frac{1}{2}z^2}\text{d}z \Big]\\
    &= \textstyle\D F\Big[ \frac{\ln\left(\frac{\D}{P_g}\right)+\frac{\sigma_g^2}{2}\tau}{\sigma_g \sqrt{\tau}}\Big],\nonumber
\end{align}
where, $F(.)$ is the Cumulative Distribution Function of the standard Gaussian random variable $\sim \mathcal{N}(0,1)$. Similarly,
\begin{align}\label{eq:pde_i2}
    I_2 &= \textstyle\frac{1}{\sigma_g \sqrt{2\pi \tau}}\Big[\int_{-\infty}^{\ln(\D)} e^{-\frac{1}{2}\big[ \frac{(y+\frac{\sigma_g^2}{2}\tau-x)}{\sigma_g\sqrt{\tau}}\big]^2}e^y\text{d}y\Big]\nonumber\\
        &= \textstyle\frac{e^x}{\sqrt{2\pi}} \int_{-\infty}^{\frac{\ln(\D)+\frac{\sigma_g^2}{2}\tau - x}{\sigma_g \sqrt{\tau}}} e^{-\frac{1}{2}\left[ z - \sigma_g \sqrt{\tau}\right]^2}\text{d}z.
\end{align}
Let, $w = z-\sigma_g \sqrt{\tau}, \implies dw = dz$. Therefore, from (\ref{eq:pde_i2}),
\begin{align}
     I_2 &= \textstyle\frac{e^x}{\sqrt{2\pi}} \int_{-\infty}^{\frac{\ln(\D)+\frac{\sigma_g^2}{2}\tau - x - \sigma_g^2 \tau}{\sigma_g \sqrt{\tau}}} e^{-\frac{1}{2}w^2}\text{d}w \nonumber\\
        &= \textstyle P_gF\Big[\frac{\ln\left(\frac{\D}{P_g}\right)-\frac{\sigma_g^2}{2}\tau}{\sigma_g \sqrt{\tau}}\Big],\nonumber
\end{align}
where, $F(.)$ is the Cumulative Distribution Function of the standard Gaussian random variable $\sim \mathcal{N}(0,1)$. Therefore, substituting the value of $\tau$,
\begin{align}\label{eq:pde_sol}
    V(P_g(t), t) &= \textstyle(\D)F\Big[ \frac{\ln\left(\frac{\D}{P_g(t)}\right)+\frac{\sigma_g^2}{2}(T_f - t)}{\sigma_g \sqrt{(T_f - t)}}\Big] \nonumber \\ &-  \textstyle P_g(t)F\Big[\frac{\ln\left(\frac{\D}{P_g(t)}\right)-\frac{\sigma_g^2}{2}(T_f - t)}{\sigma_g \sqrt{(T_f - t)}}\Big].
\end{align}
Comparing equation (\ref{eq:pde_sol}) with equation (\ref{eq:portfolio}),
\begin{align}\label{eq:batt_policy}
    a(t) &=\textstyle -F\Big[\frac{\ln\left(\frac{\D}{P_g(t)}\right)-\frac{\sigma_g^2}{2}(T_f - t)}{\sigma_g \sqrt{(T_f - t)}}\Big],\\ \nonumber
    b(t) &= \textstyle\frac{\D}{P_b}F\Big[ \frac{\ln\left(\frac{\D}{P_g(t)}\right)+\frac{\sigma_g^2}{2}(T_f - t)}{\sigma_g \sqrt{(T_f - t)}}\Big].
\end{align}
This completes the proof.

\bibliography{references}
\end{document}